\documentclass[a4paper,11pt]{article}
\usepackage{amssymb,amsmath,amsthm}
\usepackage{graphics,graphicx,color}
\usepackage{srcltx}
\usepackage{tcolorbox}
\DeclareGraphicsExtensions{.eps} \textwidth=169mm
\definecolor{darkred}{rgb}{0.0,0.0,2}

\newtheorem{proposition}{Proposition}
\newtheorem{theorem}{Theorem}
\newtheorem{lemma}{Lemma}

\newtheorem{remark}{Remark}
{\rm}

\def\eps{\varepsilon}

\author{A. Piatnitski$^{\small\bf a,b}$,
A. Shamaev$^{\small\bf c,d}$\thanks{This work was carried out in the framework of the agreement AAAA-A20-120011690138-6 of Institute of Problem in Mechanics
and partially supported by the Moscow Center for Fundamental and Applied Mathematics, the Ministry  of  Science of Higher education grant  № 075-15-2019-1621}\ ,
E. Zhizhina$^{\small\bf a}$ \\
\\
{
\small $^{\small\bf a}$Institute for Information Transmission Problems of RAS,}\\[-1.5mm]
{\small 19, Bolshoi Karetnyi per. build.1, 127051 Moscow, Russia}\\
{\small $^{\small\bf b}$The Arctic University of Norway, Campus Narvik,}\\[-1.5mm]
{\small P.O.Box 385, Narvik 8505, Norway}\\
{\small $^{\small\bf c}$ Lomonosov Moscow State University,}\\[-1.5mm]
{\small  GSP-1, Leninskie Gory, 119991 Moscow, Russia}\\
{\small $^{\small\bf d}$ Ishlinsky Institute for Problems in Mechanics of RAS}\\[-1.5mm]
{\small Prospekt Vernadskogo 101-1, 119526 Moscow, Russia}\\
$ $}

\def\eps{\varepsilon}

\begin{document}

\title{Mathematical multi-scale model of water purification}


\maketitle


\bigskip

\begin{abstract}
In this work  we consider a mathematical model of the water treatment process and determine the
effective characteristics of this model.
At the microscopic length scale we describe our model in terms of a lattice random walk in a high-contrast periodic medium with absorption.  Applying then the upscaling procedure we obtain the macroscopic model for total mass evolution.
We discuss both the dynamic and the stationary regimes, and show how the efficiency of the purification process depends on the characteristics of the macroscopic model.
\end{abstract}

\noindent

\section{Introduction}

The problem of water purification has great practical importance and gives rise to many interesting mathematical questions.  Mathematical modelling of water treatment has become increasingly popular in recent years, see e.g. \cite{Ref1, Ref2, Ref3}.
In the present work we deal with  mathematical models for the treatment of  wastewater in biofilm reactors  and in filters filled with granules which are made of  nano-porous super-hydrophilic materials.
We use here a combination of a probabilistic approach and a homogenization technique for modelling the purification process.

To clarify the motivation of the model we shortly describe one of the biofilters used in industrial water purification process. A biofilm reactor is a tank of cylindrical shape which is about one meter high and of diameter about 20 cm. It is packed with parallelepipeds consisting of thin pressed polymer fibers. A typical volume of such a parallelepiped is 15-20 cm$^3$, and the fibers are small rods whose length is about 1 cm. Each such a rod is covered with a thin biologically active biofilm, these biofilms being filled with bacteria for which impurities within water are a nutrition.
Water is supplied to the upper cross section of the device and then trickles down drop by drop along the rods so that the biofilms
covering the rods are getting wet.  The material of biofilms is designed in such a way that its diffusion coefficient  is much smaller then that
in the surrounding fluid domain.  The polluted water penetrates the biofilms and
the impurities are consumed by the bacteria. The intensity of this process depends on the concentration of both the bacteria and  the impurities at the biofilms boundary.
The averaged speed of water also influences the said intensity. 
The biofilter is efficient if the said averaged speed is sufficiently small.
All together there are several millions such rods in the device, they are called basic elements of the biofilter. We would like to construct an adequate model of the mentioned above
consumption process for one rod and then to model the whole process of water purification.  Our goal is to evaluate the drop in the water pollution level.  Several models of this type have been considered in a number of works, in particular in \cite{[10],[14]}.

In \cite{[14]} the consumption of impurities in one basic element is described by  a system of differential equations including a diffusion equation
in 3D  cylindrical domain and a transport equation at the rod border. This problem does not have an analytic solution. So it is natural to simulate
its solution numerically. To this end in \cite{[14]} the whole cylindrical tank is divided into layers in the vertical direction and calculate numerically
the drop of the impurity concentration at each layer.
Also, in \cite{[14]} the asymptotic analysis of the system is performed provided a small thickness of the rods.

We turn to another model of water treatment.
 One of the most common pollutants of the wastewater are petrol and oils impurities, and the water purification from oil and petrol products refer to the highly important environmental problems. One of the modern methods of water purification is described in \cite{KK}. It is the filtering method using granules made of innovative nano-porous super-hydrophilic materials.
For the practical implementation a design of the classical pressure filter for granulated filter bed was chosen. The system is represented by a vertical cylindrical filter with a distribution system below and above. The filter is filled with granulated bed of the grade 0.7 - 1.7 mm, preliminary impregnated with water. The filter has a height of 1.5 m and an average pore diameter of 6.5 nm.  The diffusion coefficient in the granules 
is much smaller than that in the surrounding solute.

In the present work we suggest a mathematical model based on a probabilistic interpretation of the water purification process described in \cite{KK}.
It is assumed that the movement of the impurities at the micro scale is described in terms of a Markov process. Namely, the  impurities can enter the porous granules with a positive probability and then either be absorbed there
or leave.

Also, it is  assumed that the basic purification elements are located periodically.
   We then divide each periodicity cell into a finite number of cubes,
   introduce the lattice formed by the centers of these cubes and
   perform the corresponding discretisation of  the Markov process.
   For the obtained random walk we define the transition probabilities
   between the sites of the same cell or neighbouring cells. This yields the description
   of the model at the microscopic level.

Our goal is to provide the macroscopic description of this process
based on upscaling  procedure. It will be shown that the
coefficients involved in the macroscopic model, i.e. the
effective characteristics of the water purification process,
 can be expressed  through  the characteristics of the model at the microscopic scale
 by means of solving a system of linear algebraic equations. The size of the system only depends
 on the number of points in the period.

The main characteristics of the quality of water purification is the exponent
that specifies the rate of decay of impurities concentration depending on the distance
to the upper cross section of the filter. In this work we provide some examples of
calculating such exponent.

The advantage of the proposed model is that it can be easily adapted to any geometry of absorbing films.
The model considered in this work can be used for  better understanding  complex treatment systems as well as for optimising the parameters of water purification devices in accordance with
the restrictions on the device productivity and the purification quality.

Let us formulate the mathematical problem underlying the considered model.
In Sect. 2 we introduce the random walk $\widehat X(n)$ on $\mathbb Z^d$, $d\geq 1$ in a periodic high-contrast
medium that models the process of water purification on a microscopic scale. We consider the random walk as a discretization of the diffusion of impurities within a device with appropriate restrictions on their movement.
In the microscopic model considered
in this work, in addition to the random walk there is also a partial absorption at the astral sites.
To describe the absorption process we modify the random walk model adding the absorbing state $\{ \star \}$. Thus, our model at the micro level is the random walk with absorption, we call this process $\mathcal{X}(t)$.
In Sect. 3 we study the large time behaviour of this process under the upscaling procedure.
To do this we assume that the transition probabilities of the random walk depend on a small parameter $\varepsilon>0$, and study the limit behaviour of the rescaled process
$\mathcal{X}_\varepsilon(t)$, as $\varepsilon\to 0$.
It turns out that there exists a nice and very useful description of the limit process as a two component continuous time Markov process $\mathfrak{X}(t) = (\mathcal{X}(t), k(t))$. Its  first component ${\mathcal{X}}(t)$  evolves in the space $\mathbb R^d \cup \{ \star \}$, while the second component is a jump Markov process $k(t)$ with a finite number of states.

In Sect. 4 we provide an example of the macroscopic (effective) model, both in dynamic and in stationary regimes. In Sect. 5 one can find the derivation of formulae for the effective characteristics of the macroscopic model. In Sect. 6 we calculate the effective matrix and the effective drift for one example, and show the connection between models at micro and macro scales.

The mathematical background of the present work has been  developed in \cite{PZh}, where we studied a symmetric random walk in high-contrast medium and constructed the limit process on the extended state space. In this work, we supplement the symmetric random walk with an additional drift and absorption, which leads to a significant modification of our previous scheme. A crucial step in our approach is constructing several periodic correctors which are introduced as solutions of auxiliary difference elliptic equations on the period. Earlier the corrector techniques in the discrete framework have been developed in \cite{Ko86} for proving the homogenization results for uniformly elliptic difference schemes.

\medskip




Various phenomena in media with a high-contrast microstructure have been widely studied by the specialists
in applied sciences and then since '90th high-contrast homogenization problems have been attracting the attention of mathematicians.
Homogenization problems for partial differential equations describing high-contrast periodic media have
been intensively investigated in the existing mathematical literature. In the pioneer work \cite{ADH90} a parabolic equation with high-contrast periodic coefficients has been considered. It was shown that the effective equation contains a non-local in time term which represents the memory effect. In the literature on porous media these models are usually called
double porosity models. Later on in \cite{Ale92}, with the help of two-scale convergence techniques, it was proved that the solutions of the original parabolic equations two-scale converge to a function which depends both on slow and fast variables.



\section{Microscale description. High-contrast discrete models.}


In this section we provide a detailed description of the random walk.  Given a probability space $(\Omega,\mathcal{F}.\mathbf{P})$,
we consider a random walk $\widehat X(n)$ in $\mathbb Z^d, \ d \ge 1$, with the transition probabilities
$p(x,y)=\Pr (x \to y)$,  $(x,y) \in \mathbb Z^d \times  \mathbb Z^d$:
\begin{equation}\label{p}
\sum_{y \in \mathbb Z^d} p(x,y)=1 \quad  \forall x \in \mathbb Z^d.
\end{equation}
Denote the transition matrix of the random walk by $P =\{ p(x,y), \ x,y \in \mathbb Z^d \}$.
We assume that the random walk satisfies the following properties:
\begin{itemize}
\item [-] {\it Periodicity}. The functions $p(x,x+\xi)$ are periodic in $x$  with a period $Y$ for all $\xi\in\mathbb Z^d$. In what follows we identify the period $Y$ with the corresponding $d$-dimensional discrete torus $\mathbb T^d$.
\item [-] {\it Finite range of interactions}. There exists $c>0$ such that
\begin{equation}\label{c1}
p(x,x+\xi)=0, \quad \hbox{if }|\xi|>c.
\end{equation}
\item [-] {\it Irreducibility.}  The random walk is irreducible in $\mathbb Z^d$.
\end{itemize}

\medskip
In this paper we consider a family of transition probabilities $p^{(\varepsilon)}(x,y)$ that satisfy all above properties and depend on a small parameter $\varepsilon>0$.
The transition probabilities $p^{(\varepsilon)}(x,y)$  describe  the so-called high-contrast periodic structure of the environment.  We suppose that the transition matrix $P^{(\varepsilon)}$ is a small perturbation
of a fixed transition matrix $P^0 = \{ p_0 (x,y) \}$ that corresponds to a symmetric random walk, i.e.
\begin{equation}\label{p0-sym}
p_0(x,y) = p_0(y,x), \quad (x,y) \in \mathbb Z^d \times  \mathbb Z^d, \quad \sum_{y \in \mathbb Z^d} p_0(x,y)=1 \quad  \forall x \in \mathbb Z^d.
\end{equation}
We say that $y \sim x, \; x,y \in \mathbb Z^d$, if $p_0 (x,y) \neq 0$. Let $\Lambda_x$ be a finite set  of $\xi\in \mathbb Z^d$ such that $x+\xi \sim x$.
We will use further the notation
\begin{equation}\label{p0-sym-bis}
p_0 (x,y) = p_0 (x, x+ \xi) = p_\xi(x) \; \mbox{ for all } \; x \sim y, \; x,y \in \mathbb Z^d, \; \mbox{ with } \;  y =x+ \xi.
\end{equation}
Thus the normalization condition can be rewritten as
$$
\sum_{\xi \in \Lambda_x} p_\xi (x)=1.
$$

The transition matrix $P^{(\varepsilon)}$ has the following form
\begin{equation}\label{PV}
P^{(\varepsilon)} = P^0 + \varepsilon D + \varepsilon^2 V.
\end{equation}
In order to characterize the matrices $P^0$, $D$ and $V$ we divide the periodicity cell into
two sets
\begin{equation}\label{defAB}
\mathbb T^d = A \cup B; \quad A,\, B \neq \emptyset, \; A \cap B =   \emptyset,
\end{equation}
and assume that $B \subset \mathbb T^d $ is a connected set such that its periodic extension denoted $B^{\sharp}$
is unbounded and connected. Here the connectedness is understood in terms of the transition matrix $P^0$. Two points $x',\,x''\in\mathbb Z^d$ are called connected if there exists a path $x^1,\ldots,x^L$ in
 $\mathbb Z^d$ such that $x^1=x'$, $x^L=x''$ and $p_0(x^j,x^{j+1})>0$ for all $j=1,\ldots, L-1$.
As a consequence we get that
\begin{equation}\label{irrB}
P^0 \; \mbox{is irreducible on } \; B^\sharp.
\end{equation}
We also denote by $A^\sharp$ the periodic extension of $A$. Then
$\mathbb Z^d = A^{\sharp} \cup B^{\sharp}$.

\bigskip
In addition to the general assumptions (\ref{p0-sym}) and \eqref{irrB} the following conditions on  the matrices $P^0$, $D$ and $V$ are imposed:
\begin{itemize}
\item[\bf --] $p_0(x,x) = 1$, if $x \in A^\sharp$;
%
\item[\bf --] $p_0(x,y) = 0$, if $x, y \in A^\sharp, \ x \neq y$;
\item[\bf --] $p_0(x,y) = 0$, if $x \in B^\sharp, \ y \in A^\sharp$;
\item[\bf --]   $d(x,y)=0$, if at least one of $x $ or $ y \in A^\sharp$;
\item[\bf --]   $v(x,y)=0$, if $x,\,y\in B^\sharp$, $x\not= y$;
\item[{\bf --}]  the elements of matrices $D$ and $V$ satisfy the relation
\begin{equation}\label{v0}
\sum_{y \in \mathbb Z^d} d (x,y) = 0, \quad
\sum_{y \in \mathbb Z^d} v (x,y) = 0 \quad \forall x \in \mathbb Z^d.
\end{equation}
\end{itemize}
\begin{remark}
In particular, the above conditions imply that $v(x,y) \ge 0$, if  at least one of $x $ or $ y \in A^\sharp$ and $x \neq y$.
\end{remark}
From the periodicity of $D$ and $V$ it also follows that
$$
d_{max}:=\max_{x,y \in \mathbb Z^d} |d(x,y)| <\infty, \quad v_{max}:=\max_{x,y \in \mathbb Z^d} |v(x,y)| <\infty.
$$
Summarizing all above conditions, we conclude that the non-zero transition probabilities defined by \eqref{PV} have the following structure:
\begin{itemize}
\item[\bf --] $p(x,y) = p_0(x,y) + O(\eps)$ with $p_0(x,y) \asymp 1$, when $x,y \in B^\sharp$ (rapid movement);
\item[\bf --] $p(x,x)=  1+O(\eps^2)$, when $x\in A^\sharp$ (slow movement);
\item[\bf --] $p(x,y) \asymp\varepsilon^2$, when $x, y \in A^\sharp, \ x \neq y$ (slow movement);
\item[\bf --] $p(x,y) \asymp\varepsilon^2$, when $x \in B^\sharp, \ y \in A^\sharp$ (rare exchange between $A^\sharp$ and $B^\sharp$).
\end{itemize}
The above choice of the transition probabilities reflects a
slow drift (of the order $\eps$) given by matrix $D$ in the fast component, and also a significant slowdown (of the order $\eps^2$) of the random walk inside the slow component.
\medskip

Further, we add to the above random walk an absorption process consistent with the structure of the periodic environment, assuming that the absorption occurs only inside the inclusions $A^\sharp$. For the description of the complete process we will denote by  $ \mathbb{S} = \mathbb{Z}^d \cup \{ \star \}$ the state space of the new process, where $\{\star \}$ is the absorption state. Then the transition matrix of the complete process with absorption has the following form
\begin{equation}\label{PVW}
Q^{(\varepsilon)} = P^{(\varepsilon)} + \varepsilon^2 W = \big( P^0 + \varepsilon D + \varepsilon^2 V \big) + \varepsilon^2 W,
\end{equation}
where $Q^{(\varepsilon)}(\star, \star) = 1, \; W(x, \star) = m>0$ and $W(x,x) = -m$ for all $x \in A^\sharp$, otherwise $W(x,y)=0$.

Let $l_0^{\infty} (\mathbb Z^d)$ be the Banach space of bounded functions on $\mathbb Z^d$ vanishing at infinity with the norm $\|f \| = \sup_{x \in \mathbb Z^d} |f(x)|$. Similarly, we consider the Banach space of bounded functions $l_0^{\infty} (\mathbb{S}) = l_0^{\infty} (\mathbb Z^d) \oplus \mathbb{C} $.

We note that random walks with transition probabilities of the form 
\begin{equation}
P^{(\varepsilon)} = P^0 +  \varepsilon^2 V
\end{equation}
has
have  been studied in  \cite{PZh}. In the present  work we supplement the model with drift and absorption. Our goal is to derive the effective evolution equation under the diffusive scaling.

\section{Upscaling}

\subsection{Rescaled process}

In what follows we study the scaling limit of the random walk on $\mathbb{S}$ with transition matrix $Q^{(\varepsilon)}$ and use $\varepsilon$  as the  scaling factor.
Denote $\varepsilon \mathbb Z^d = \{z: \frac{z}{\varepsilon} \in \mathbb Z^d \}$, then $\varepsilon \mathbb Z^d = \varepsilon A^{\sharp} \cup \varepsilon B^{\sharp}$, and let $\varepsilon \mathbb{S} = \varepsilon \mathbb{Z}^d \cup \{ \star \}$. In what follows the symbols $x$ and $y$
are used for the variables on $\mathbb Z^d$ (fast variables), while the symbols $z$ and $w$ for the variables on $\eps \mathbb Z^d$ (slow variables). Notice that the state $\{ \star \}$ does not change under the scaling.

We introduce now the rescaled process.
Denote by $T_{\varepsilon}$ the transition operator associated with the transition matrix \eqref{PVW}:
\begin{equation}\label{Peps}
T_{\varepsilon} f(z) = \sum_{w \in \varepsilon \mathbb{S}} q_\varepsilon (z,w) f(w) =  \sum_{w \in \varepsilon \mathbb Z^d} q_\varepsilon (z,w) f(w) + q_\varepsilon (z, \star) f(\star), \quad f \in  l_0^{\infty} (\varepsilon \mathbb{S}),
\end{equation}
where $ q_\varepsilon (z,w)$  are elements of the matrix $ Q^{(\varepsilon)}$, see \eqref{PVW}. Namely, $q_\varepsilon (z,w) = p_\varepsilon (z,w) = p(\frac{z}{\varepsilon}, \frac{w}{\varepsilon})$, when $z,w  \in \varepsilon \mathbb Z^d $, where $p_\varepsilon(z, w)$ are elements of the matrix $P^{(\varepsilon)}$ defined in \eqref{PVW};
$ q_\varepsilon (z, \star) = \varepsilon^2 m$, if $z \in  \varepsilon A^{\sharp}$ and $ q_\varepsilon (z, \star) = 0$, if $z \in  \varepsilon B^{\sharp}$. Then the operator
\begin{equation}\label{L_e}
L_\varepsilon \ = \ \frac{1}{\varepsilon^2} (T_{\varepsilon}-I)
\end{equation}
is the difference generator of the rescaled process ${\mathcal{X}}_\varepsilon (t)$
on $\varepsilon \mathbb{S} = \varepsilon \mathbb{Z}^d \cup \{ \star \}$  with transition operator $T_\varepsilon$. The rescaled process has two components:
\begin{equation}\label{2comp}
 {\mathcal{X}}_\varepsilon (t) = \{  \widehat X_\varepsilon(t), \ \widehat s(t)   \},
\end{equation}
where
$ \widehat X_\varepsilon(t) = \varepsilon \widehat X(\left[\frac{t}{\varepsilon^2}\right])$ is the rescaled random walk on
$\varepsilon \mathbb Z^d $, and the latter component $\widehat s(t)$ lives on $\{ \star \}$.

The goal of the paper is to describe the limit behavior of the rescaled process ${\mathcal{X}}_\varepsilon (t)$, as $\eps\to 0$, to construct the limit process, and to find the explicit expressions for all effective characteristics of the limit process.

\subsection{Extended random walk}

{
Homogenization of non-stationary processes in high contrast environments often results in the effective equations with nonlocal in time terms
representing the memory effect. 
As was shown in  \cite{PZh} the limit process for a random walk in a high contrast environment remains Markov   if we equip the original random walk with additional component(s) and consider the obtained random walk in the extended state space.

In this subsection we describe the constructions of an extended random walk introduced in  \cite{PZh}.}
We equip the random walk  $\widehat X_\varepsilon (t)$ (the first component in \eqref{2comp}) with an additional component(s) in the same way as it has been done in \cite{PZh}.
  Assume that the set $A$ defined in \eqref{defAB} contains $M \in \mathbb N$ sites of $\mathbb T^d$: $A = \{ x_1, \ldots, x_M \}$.
For each $k=1, \ldots, M$ we denote by $\{x_k \}^{\sharp}$ the periodic extension of the point $x_k \in A$, then
\begin{equation}\label{decZ}
\varepsilon \mathbb Z^d  = \varepsilon B^{\sharp} \cup \varepsilon A^{\sharp} = \varepsilon B^{\sharp} \cup \varepsilon \{ x_1 \}^{\sharp} \cup \ldots  \cup \varepsilon \{ x_M \}^{\sharp}.
\end{equation}
We assign to each $z \in \varepsilon \mathbb Z^d$ the index $k(z) \in \{0,1, \ldots, M\}$ depending on the component in decomposition (\ref{decZ}) to which $z$ belongs:
\begin{equation}\label{kx}
k(z) \ = \ \left\{
\begin{array}{l}
0, \; \mbox{ if } \;  z \in \varepsilon B^{\sharp}; \\
j, \;  \mbox{ if } \;  z \in \varepsilon \{ x_j \}^{\sharp}, \; j= 1, \ldots, M.
\end{array}
\right.
\end{equation}
With this construction in hands we introduce the metric space
\begin{equation}\label{Ee}
E_\varepsilon\ = \  \left\{ (z, k(z)), \; z \in \varepsilon \mathbb Z^d,\; k(z) \in \{0,1,\ldots, M \} \right\}, \quad E_\varepsilon \subset \varepsilon \mathbb Z^d \times \{0,1,\ldots, M \}
 \end{equation}
 with a metric that coincides with the metric in  $\varepsilon \mathbb Z^d$ for the first component of $(z,k(z)) \in E_\varepsilon$.

The  index  $k(\star) = \star$ is assigned to the state $z =\{ \star\}$. Thus the extended version of the absorption state is $\{ \star, \star \}$, but for simplicity we will keep the notation $\{\star \}$. Denote by $\mathbb{S}_{E_\varepsilon} = E_\varepsilon \cup \{ \star \}$, and
in what follows instead of $\mathcal{X}_\varepsilon (t)$ we will consider the process
$$
{\mathfrak{X}}_\varepsilon (t) = \{ \mathcal{X}_\varepsilon (t), k(t) \},
\quad k(t) \in  \{ 0,1, \ldots, M, \star \}.
$$

We denote the space of bounded functions on $\mathbb{S}_{E_\varepsilon}$ by ${\cal B}(\mathbb{S}_{E_\varepsilon})$ and construct the transition operator $T_\varepsilon$ of the process ${\mathfrak{X}}_\varepsilon (t)$ on $\mathbb{S}_{E_\varepsilon}$ using the same transition probabilities as in operator (\ref{Peps}):
\begin{equation}\label{Te}
\begin{array}{l}
\displaystyle
(T_{\varepsilon} f) (z, k(z)) = \sum_{w \in \varepsilon \mathbb{S}} q_\varepsilon (z,w) f(w, k(w)) =  \sum_{w \in \varepsilon \mathbb Z^d} q_\varepsilon (z,w) f(w, k(w)) + q_\varepsilon (z, \star) f(\star), \\[6mm]
\displaystyle
(T_{\varepsilon} f) (\star) =  f(\star), \qquad f \in  {\cal B} (\mathbb{S}_{E_\varepsilon})
\end{array}
\end{equation}
with $\{ q_\varepsilon (z, w) \} = Q^{(\varepsilon)}$.

Then $T_\varepsilon$ is the contraction on  ${\cal B}(\mathbb{S}_{E_\varepsilon})$:
$$
\| T_\varepsilon f \|_{ {\cal B}(\mathbb{S}_{E_\varepsilon})} \le \sup_{(z, k(z)) \in \mathbb{S}_{E_\varepsilon}} |f (z, k(z))|,  \quad f \in {\cal B}(\mathbb{S}_{E_\varepsilon}).
$$
\begin{remark}\label{RT}
Since the point  $(z,k(z)) \in E_\varepsilon$ is uniquely defined by its first coordinate  $z \in \varepsilon \mathbb Z^d$, then we can use  $z \in \varepsilon \mathbb Z^d$ as a coordinate in $E_\varepsilon$ (considering $E_\varepsilon$ as a graph of the mapping $k: \varepsilon \mathbb Z^d \to \{0,1,\ldots, M \}$). In particular, for the transition probabilities of the random walk on $E_\varepsilon$ we keep the same notations $q_\varepsilon(z,w)$ as in \eqref{Peps}.
\end{remark}

\subsection{Limit process}

In this subsection, we construct a limit process, which is a Markov process completely determined by its generator.
We denote $E= \mathbb R^d \times \{0 ,1, \ldots, M\}$, and $ C_0 (E)$ stands for the Banach space of continuous  functions vanishing at infinity.
Together with $E$ we consider $\mathbb{S}_E = E \ \cup \ \{\star\}$ and denote $ C_0 (\mathbb{S}_E) = C_0(E) \oplus \mathbb{C}$. Then  $F \in C_0(\mathbb{S}_E)$ can be represented as
$$
F=(F(z,k), F(\star)), \quad \mbox{where } \; F(z,k) = \{ f_k (z) \in C_0(\mathbb R^d), \; k=0,1,\ldots, M \}, \quad F(\star) \in \mathbb{C},
$$
and the norm in  $ C_0 (\mathbb{S}_E)$ is equal to
$$
\|F\|_{C_0 (\mathbb{S}_E)} = \max \big\{ \|F (z,k) \|_{C_0(E)}, \ F(\star)  \big\},
$$
where
$$
\|F\|_{C_0(E)} = \max_{k = 0, 1 \ldots, M}  \|f_k\|_{C_0(\mathbb R^d)}.
$$

Consider the operator
\begin{equation}\label{LM}
L F(z,k) =
\Big( \Theta \cdot \nabla \nabla f_0 (z) + b \cdot \nabla f_0 (z) \Big) {\bf 1}_{\{k=0\}}
 + \ L_A F(z,k), \qquad L F(\star) =0,
\end{equation}
where ${\bf 1}_{\{k=0\}}$ is the indicator function, $\Theta \cdot \nabla \nabla f_0=\mathrm{Tr}(\Theta \nabla \nabla f_0 )$,  $\Theta$ is  a positive definite  matrix defined below in \eqref{theta},
$b$ is a vector of the effective drift defined also below by \eqref{b}. Both of these effective characteristics of the limiting process are written in terms of the first corrector of the corresponding problem on the cell.
The operator  $L_A$ is a generator of a Markov jump process
\begin{equation}\label{LA}
L_A F(z,k) \ = \  \sum_{{j=0}\atop{j \neq k}}^M \alpha_{kj} (f_j(z) - f_k(z)) \ {+ \ m \big( F(\star) - f_k(z) \big) {\bf 1}_{\{k \neq 0\}} }
\end{equation}
with
\begin{equation}\label{alpha_prior}
\begin{array}{c}
\displaystyle
\alpha_{0j} \ = \ \frac{1}{|B|} \ \sum_{x\in B} \ \ \sum_{y\in \{x_j\}^\sharp}v(x,y),\quad
\alpha_{j0} = \sum_{x \in B^\sharp} v(x_j, x), \quad  j=1,\ldots,M,\\[7mm]
\displaystyle
\alpha_{kj} =  \sum_{y\in \{x_j\}^\sharp}v(x_k, y), \ \ j,\,k=1,\ldots,M,\ j\not=k,
\end{array}
\end{equation}
Notice that the parameters $\alpha_{jk}$, $j,\,k=0,1,\ldots, M$, are non-negative and define intensities of the limit Markov jump process on the period $Y$.



The operator $L$ is defined on the core
\begin{equation}\label{core}
D \ = \ \{ (f_0, f_1, \ldots, f_M), \; f_0 \in C_0^{\infty}(\mathbb R^d), \; f_j \in C_0(\mathbb R^d), \; j=1, \ldots, M\} \oplus \mathbb{C}  \ \subset \ C_0(\mathbb{S}_E),
\end{equation}
and
\begin{equation}\label{coreE}
D_E = \{ (f_0, f_1, \ldots, f_M), \; f_0 \in C_0^{\infty}(\mathbb R^d), \; f_j \in C_0(\mathbb R^d), \; j=1, \ldots, M\}
\end{equation}
is a dense set in $ C_0 (E)$.
One can check that the operator $L$ on $C_0 (\mathbb{S}_E)$ satisfies the positive maximum principle, i.e. if
$
F \in C_0(\mathbb{S}_E)$  and  $\max_{ E\cup \{\star \} } F(z,k) =  F(z_0, k_0),
$
then  $L F (z_0, k_0) \le 0$.
Since $L_A$ is a bounded operator in $C_0(\mathbb{S}_E)$, the operator $\lambda-L$ is invertible for sufficiently large $\lambda$.
Then by the Hille-Yosida theorem the closure of $L$ is a generator of a strongly continuous, positive, contraction semigroup $T(t)$  on $ C_0 (\mathbb{S}_E)$.

{ 
Let us describe the limit process $\mathfrak{X}(t)$ generated by the operator $L$.  It is a two component continuous time Markov process $\mathfrak{X}(t) = \{ {\mathcal{X}}(t), k(t)\}$, where  the first component ${\mathcal{X}}(t)$  lives in $\mathbb R^d \cup \{\star\}$,  the second component is a continuous time jump Markov process $k(t)$ on the state space $K=\{0,1,2, \ldots, M, \star\}$.
The process $k(t)$  does not depend on the other components;  its transition rates $\alpha_{ij}$
are expressed in terms of the transition
probabilities of the original random walk, see \eqref{alpha_prior}.
The probability of jump between any two states $i,j \in \{0,1,2, \ldots, M\}, \  i \neq j $, is equal to $\alpha_{ij}$. The absorbing state $\{ \star \}$ is reachable only from the "astral" states $\{ 1,2,\ldots,M \}$ with the same intencity $m$.
Thus, the matrix corresponding to the generator $L_A$ has the following form
\begin{equation}\label{matrix}
\left(
\begin{array}{cccccc}
- \sum_{j=1}^M \alpha_{0j} & \alpha_{01} &  \alpha_{02} & \ldots & \alpha_{0 M} & 0 \\
 \alpha_{10} & - \sum_{j=0, j\neq 1}^M \alpha_{1j} - m &  \alpha_{12} & \ldots & \alpha_{1 M} & m \\
\vdots & \vdots & \vdots & \ddots & \vdots & \vdots \\
 \alpha_{M 0} &  \alpha_{M 1} & \alpha_{M2} & \ldots & - \sum_{j=0}^{M-1} \alpha_{Mj} - m & m \\
0 & 0 & 0 & \ldots & 0 & 0
\end{array}
\right)
\end{equation}
 When $k(t) = 0$,  the first
component ${\mathcal{X}}(t)$ evolves along the trajectories of a diffusion process in $\mathbb R^d$ with the corresponding effective characteristics, while when $k(t) \neq 0$,  the first
component $\mathcal{X}(t)$ remains still until $k(t)$ takes again the value $0$.
Thus the trajectories of  ${\mathcal{X}}(t)$  coincide with  the trajectories of a diffusion process in $\mathbb R^d$ on those time intervals where  $k(t)=0$.  As long as $ k(t) \neq 0 $ the first component ${\mathcal{X}}(t)$ does
not move, and only the second component $k(t)$ of the process evolves.
Additionally,  the process $k(t)$ can jump from the astral states $\{ 1, \ldots, M \}$ to the absorbing state $\{ \star \}$ with intensity $m$, and upon reaching this state, the process    never  leaves it.
}

\subsection{Main result. The convergence of semigroups}


In this subsection we formulate the main result of this work on  convergence (upscaling) to the limit process constructed in the previous subsection.

Let $l_0^{\infty}(E_\varepsilon)$ be a Banach space of functions on $E_\varepsilon$ vanishing as $|z| \to \infty$ with the norm
\begin{equation}\label{normle}
\|f\|_{l_0^{\infty}(E_\varepsilon)} = \sup_{(z,k(z)) \in E_\varepsilon}|f(z,k(z))| = \sup_{z \in \varepsilon \mathbb Z^d}|f(z,k(z))|,
\end{equation}
and denote $l_0^{\infty}(\mathbb{S}_{E_\varepsilon}) = l_0^{\infty}(E_\varepsilon) \oplus \mathbb{C}$.
For every $F \in C_0(\mathbb{S}_E)$ we define
the function $\pi_\varepsilon F \in l_0^{\infty}(\mathbb{S}_{E_\varepsilon}) $ as follows:
\begin{equation}\label{pi-eps}
(\pi_\varepsilon  F) (z, k(z)) \ = \ \left\{
\begin{array}{ll}
f_0 (z), & \mbox{if} \; z \in \varepsilon B^{\sharp}, \; k(z) =0; \\
f_1 (z), & \mbox{if} \; z \in \varepsilon \{ z_1 \}^{\sharp}, \; k(z) =1;\\
\cdots \\
f_M (z), & \mbox{if} \; z \in \varepsilon \{ z_M \}^{\sharp}, \; k(z) =M,
\end{array}
\right.
\end{equation}
and $ \pi_\varepsilon  F(\star) = F(\star)$.
Then $\pi_\varepsilon$ defines a bounded linear transformation $ \pi_\varepsilon: C_0(\mathbb{S}_E) \to l_0^{\infty}(\mathbb{S}_{E_\varepsilon})$.

\begin{theorem}\label{T1}
Let $T(t)$ be a strongly continuous, positive, contraction semigroup on $ C_0 (\mathbb{S}_E)$ with generator $L$ defined by \eqref{LM}--\eqref{LA}, and $T_\varepsilon$ be the linear operator on $l_0^{\infty}(\mathbb{S}_{E_\varepsilon})$ defined by \eqref{Te}.

Then for every $F \in C_0(\mathbb{S}_E)  $
\begin{equation}\label{M-astral}
T_\varepsilon^{\left[ \frac{t}{\varepsilon^2} \right]} \pi_\varepsilon F \ \to \  T(t) F   \quad \mbox{for all} \quad t \ge 0
\end{equation}
as $\varepsilon \to 0$.
\end{theorem}

\begin{proof}

The proof of (\ref{M-astral}) relies on the approximation techniques from \cite{EK} used for the proof of  convergence of semigroups.
According to results of \cite[Theorem 6.5, Ch.1]{EK}  the semigroups convergence stated in \eqref{M-astral} is equivalent to the statement  which is the subject of the  next lemma.

\begin{lemma}\label{Fn}
For every $F \in D$, where
$D$ was defined by \eqref{core}, there exists $F_\varepsilon \in  l_0^{\infty}(\mathbb{S}_{E_\varepsilon})$ such that
\begin{equation}\label{F0}
 \| F_\varepsilon - \pi_\varepsilon F\|_{ l_0^{\infty}(\mathbb{S}_{E_\varepsilon})} \to 0
\end{equation}
and
\begin{equation}\label{F1}
\|L_\varepsilon F_\varepsilon - \pi_\varepsilon LF\|_{ l_0^{\infty}(\mathbb{S}_{E_\varepsilon} )} \to 0 \quad \mbox{as } \; \varepsilon \to 0.
\end{equation}
\end{lemma}

\begin{proof}
For every $F = \big((f_0, f_1, \ldots, f_M), F(\star) \big) \in D$ with $(f_0, f_1, \ldots, f_M) \in D_E$ we present a function $F_\varepsilon$ that guarantees the convergence \eqref{F0}-\eqref{F1}.
Let us take $F_\varepsilon \in l_0^{\infty}(\mathbb{S}_{E_\varepsilon})$ in the following form
\begin{equation}\label{4}
F_\varepsilon (z, k(z)) \ = \ \left\{
\begin{array}{ll}
f_0 (z) + \varepsilon (\nabla f_0 (z), h(\frac{z}{\varepsilon})) +  \varepsilon^2 (\nabla \nabla f_0 (z), g(\frac{z}{\varepsilon}))  \\[3mm] +\varepsilon^2  \sum_{j=1}^M q_j (\frac{z}{\varepsilon}) (f_0(z) - f_j (z)),  & \mbox{if} \;  z \in \varepsilon B^{\sharp}, \; k(z) = 0, \\ \\
f_1 (z),  & \mbox{if} \;  z \in \varepsilon \{ x_1 \}^{\sharp}, \; k(z) = 1,\\ \cdots & \cdots \\
f_M (z),  & \mbox{if} \;  z \in \varepsilon \{ x_M \}^{\sharp}, \; k(z) =M,
\end{array}
\right.
\end{equation}
and $F_\varepsilon(\star) = F(\star)$. Here $ h(y),  g(y), q_j (y), j = 1, \ldots, M,$ are periodic bounded functions that will be defined below. The boundedness together with (\ref{4}) immediately imply that
$$
\| F_\varepsilon - \pi_\varepsilon F\|_{ l_0^{\infty}(\mathbb{S}_{E_\varepsilon})} \| =
\sup_{z \in \varepsilon \mathbb Z^d} |  F_\varepsilon (z, k(z)) - \pi_\varepsilon F(z, k(z))  |  \to 0
$$ as $\varepsilon \to 0$. Thus, convergence (\ref{F0}) holds.
\medskip

Let us turn to the second convergence stated in (\ref{F1}). Since $(L_\varepsilon F_\varepsilon)(\star)=(LF)(\star) = 0$, it suffices to show that
\begin{equation}\label{convL}
 \|L_\varepsilon F_\varepsilon - \pi_\varepsilon LF\|_{ l_0^{\infty}(\mathbb{S}_{E_\varepsilon})} = \sup_{z \in \varepsilon \mathbb Z^d } |L_\varepsilon F_\varepsilon (z, k(z)) - \pi_\varepsilon LF (z, k(z))| \to 0.
\end{equation}
In the proof of \eqref{convL}
we use the same arguments as in the paper \cite{PZh}.
According to (\ref{Te}) and (\ref{PVW}) the operator  $L_\varepsilon $ can be written as
$$
L_\varepsilon = \frac{1}{\varepsilon^2} ( T_\varepsilon^0 + \varepsilon D_\varepsilon + \varepsilon^2 V_\varepsilon + \varepsilon^2 W_\varepsilon - I ) = L_\varepsilon^0 + V_\varepsilon + W_\varepsilon,
$$
where
\begin{equation}\label{L0-bis}
L_\varepsilon^0 = \frac{1}{\varepsilon^2} ( T_\varepsilon^0 + \varepsilon D_\varepsilon - I).
\end{equation}
We consider next separately the cases when $z \in \varepsilon B^{\sharp}$, and $z \in \varepsilon A^{\sharp}$.
Since the second component in $E_\varepsilon$ is a function of the first one, in the remaining part of the proof for brevity  write $F_\varepsilon(z)$ instead of $F_\varepsilon(z,k(z))$.

Let $z \in \varepsilon B^{\sharp}$. The first component ($k(z)=0$) of $F_\varepsilon$ in \eqref{4} can be written as a sum
\begin{equation}\label{5}
F_\varepsilon (z) = F_\varepsilon^P(z) + F_\varepsilon^Q(z), \quad  z \in \varepsilon B^{\sharp},
\end{equation}
where
\begin{equation}\label{FP}
F_\varepsilon^P(z) = f_0 (z) + \varepsilon \left(\nabla f_0 (z), h(\frac{z}{\varepsilon})\right) +  \varepsilon^2 \left(\nabla \nabla f_0 (z), g(\frac{z}{\varepsilon})\right),
\end{equation}
\begin{equation}\label{FQ}
 F_\varepsilon^Q (z) = \varepsilon^2 \sum_{j=1}^M q_j (\frac{z}{\varepsilon}) (f_0 (z) - f_j (z)).
\end{equation}
Since $W_\varepsilon F_\varepsilon (z) = 0$, if $z \in  \varepsilon B^{\sharp}$, then
\begin{equation}\label{Ldec}
L_\varepsilon F_\varepsilon = (L_\varepsilon^0 + V_\varepsilon)F_\varepsilon = L_\varepsilon^0 (F_\varepsilon^P  +  F_\varepsilon^Q) + V_\varepsilon F_\varepsilon =  L_\varepsilon^0 F_\varepsilon^P  + L_\varepsilon^0 F_\varepsilon^Q  + V_\varepsilon F_\varepsilon.
\end{equation}
To estimate
\begin{equation}\label{B-estimate}
\sup_{z \in \varepsilon B^{\sharp} } |L_\varepsilon F_\varepsilon (z) - \pi_\varepsilon LF (z)| =
\sup_{z \in \varepsilon B^{\sharp} } | L_\varepsilon^0 F_\varepsilon^P(z)  + L_\varepsilon^0 F_\varepsilon^Q (z)  + V_\varepsilon F_\varepsilon (z) - \pi_\varepsilon LF (z)|
\end{equation}
we will use two following propositions.

\begin{proposition}\label{Prop1}
There exist  bounded periodic functions $h(y)=\{h_i(y)\}_{i=1}^d$ and $g(y)=\{g_{im}(y)\}_{i,m=1}^d$ (correctors) and a positive definite matrix $\Theta>0$, such that
\begin{equation}\label{PP1}
L_\varepsilon^0 F_\varepsilon^P  \ \to \Theta \cdot \nabla \nabla f_0 + b \cdot  \nabla f_0, \quad \mbox{i.e.} \ \sup_{z \in \varepsilon B^{\sharp}}|L_\varepsilon^0 F_\varepsilon^P (z) -  \Theta \cdot \nabla \nabla f_0(z) - b \cdot  \nabla f_0(z)| \to 0  \; \mbox{ as} \; \varepsilon \to 0,
\end{equation}
where $F_\varepsilon^P$ is defined in \eqref{FP}.
\end{proposition}

The proof of this proposition is based on the corrector techniques, it is given in the Appendix.
\bigskip

\begin{proposition}\label{Prop2}
There exist bounded periodic functions $q_j(x)$,  $j = 1, \ldots, M$,  on $B^\sharp$ such that
\begin{equation}\label{7}
\sup_{z \in \varepsilon B^{\sharp}} \left| (L_\varepsilon^0 F_\varepsilon^Q + V_\varepsilon F_\varepsilon)(z) - \sum_{j=1}^M \alpha_{0j} (f_j(z) - f_0(z))  \right| \ \to \ 0  \quad \mbox{as } \; \varepsilon \to 0,
\end{equation}
where  $\alpha_{0j} > 0$ are constants defined in \eqref{alpha_prior}, and $F_\varepsilon^Q$ is introduced in \eqref{FQ}.
\end{proposition}
\noindent
The proof of Proposition \ref{Prop2} is the same as in \cite{PZh}.
We give a proof in the Appendix for complete presentation.
\bigskip

Since
\begin{equation}\label{piL-0}
\pi_\varepsilon LF (z) = \big( \Theta \cdot \nabla \nabla f_0(z) + b \cdot  \nabla f_0(z) \big) {\bf 1}_{\{k=0\}} + (L_A F)(z,0), \quad z \in  \varepsilon B^{\sharp},
\end{equation}
where
$$
(L_A F)(z,0) =  \sum_{j=1}^M \alpha_{0j} (f_j(z) - f_0(z)),
$$
then Propositions \ref{Prop1} - \ref{Prop2} together with \eqref{B-estimate} and (\ref{piL-0}) yield
\begin{equation}\label{B}
\sup_{z \in \varepsilon B^{\sharp} } |L_\varepsilon F_\varepsilon (z) - \pi_\varepsilon LF (z)|   \to 0 \quad \varepsilon \to 0.
\end{equation}



\bigskip

Next we consider the case when $z \in \varepsilon A^{\sharp}$, and prove that
\begin{equation}\label{10}
\sup_{z \in \varepsilon A^{\sharp} } |L_\varepsilon F_\varepsilon (z) - \pi_\varepsilon L F (z)|   \to 0 \quad \varepsilon \to 0.
\end{equation}
Let $z \in \varepsilon \{ x_k \}^{\sharp} \subset \varepsilon A^{\sharp}$. From (\ref{4}), \eqref{L0-bis} and continuity of functions $f_k$ it follows that
$$
(L_\varepsilon F_\varepsilon)(z)= (L_\varepsilon^0 + V_\varepsilon + W_\varepsilon) F_\varepsilon (z) = V_\varepsilon F_\varepsilon (z) + W_\varepsilon F_\varepsilon (z)=
$$
\begin{equation}\label{11}
=\sum_{{j=1}\atop{j \neq k}}^M   \sum_{y \in \{ x_j \}^\sharp} v(x_k, y) (f_j(z) - f_k(z)) +  \sum_{x \in B^\sharp} v(x_k, x) (f_0(z) - f_k(z)) + m(F(*)- f_k(z)) +o(1),
\end{equation}
as $\varepsilon \to 0$. Here we have used the fact that $f_k (z') = f_k(z) + o(1)$ when $|z-z'| \to 0$.
Recall that $x, y \in Y$ are variables on the periodicity cell, and $v(x_k, x_j)$ are the elements of the matrix $V$. On the other hand, according
\eqref{LA} and \eqref{pi-eps} $\pi_\varepsilon L F (z)$ for $z \in \varepsilon \{ x_k \}^{\sharp}$  has the following form
\begin{equation}\label{39A}
\pi_\varepsilon L F (z) = \sum_{{j=0}\atop{j \neq k}}^M \alpha_{kj} (f_j(z) - f_k(z)) + m(F(*)- f_k(z)), \quad k=1, \ldots, M,
\end{equation}
where constants $\alpha_{k0}, \ \alpha_{kj}$ are given by \eqref{alpha_prior}.
Thus,
relations (\ref{11}) and \eqref{39A} imply (\ref{10}).

Finally, (\ref{convL}) is a consequence of (\ref{B}) and (\ref{10}), and Lemma \ref{Fn} is proved.
\end{proof}

It remains to recall that \eqref{M-astral} is a straightforward consequence of the above approximation theorem. This completes the proof of Theorem \ref{T1}.
\end{proof}

\section{Dynamics of pollution. Stationary regime}

In this section we consider an example of the limit dynamics in the case when the astral set $A$ contains one point.
We also derive an equation on the first component $\rho_0(x,t)$ of the astral diffusion describing a visible dynamics of the pollution density.

Denote by
$$
\rho (x,t) = \big( \rho_0(x,t), \rho_1(x,t), \rho_2(t) \big)
$$
the three-component density of pollution, where
$\rho_0(x,t)$ is the density outside of  micro-granules,
$ \rho_1(x,t)$ is the density inside of  micro-granules,
$ \rho_2(t) $ is the density of pollution accumulated (or absorbed)  as a result of cleaning by time $t$.

The conservation principle reads
$$
\int (\rho_0(x,t) + \rho_1(x,t)) dx + \rho_2(t) \equiv const \quad \forall t.
$$

The corresponding model at microscopic scale is an one-point astral model with absorption. Then for the limit dynamics, we obtain the following evolution equations for $\rho(x,t)$
\begin{equation*}\label{1-bis}
\partial_t \rho = L^* \, \rho,
\end{equation*}
or
\begin{equation}\label{P0P1}
\left\{
\begin{array}{l}
\partial_t \rho_0 = \Theta \cdot \nabla \nabla \rho_0 -  b \cdot \nabla \rho_0 - \lambda(0) \rho_0 +  \lambda(1) \rho_1 \\[2mm]
\partial_t \rho_1 = - \big( \lambda(1) + m \big) \rho_1 +  \lambda(0) \rho_0 \\[2mm]
\partial_t \rho_2 = m \int \rho_1(x,t) \, dx.
\end{array}
\right.
\end{equation}
with initial data $\rho(x,0) = \big( \pi_0(x), \pi_1(x), \pi_2 \big)$.
Here $\Theta$, $ b$ are the effective diffusion matrix and the effective drift depending on the geometry of the micro-scale model, $\lambda(0)>0, \ \lambda(1)>0$ are the rates of exchanging between inside and outside regions:
{
$\lambda(0)$ is the intensity of the water flows into cleaning inclusions, while $\lambda(1)$ is the intensity of flows from inclusions.
All of them are the parameters of the limit model. Below in App.1-2 we will show how the effective parameters $\Theta$ and $ b$ can be found from the micro-scale model.
}

The solution of the second equation in \eqref{P0P1} has the form
\begin{equation*}\label{P1}
\rho_1(x,t) = e^{-\lambda_m t} \pi_1(x) + \lambda(0) \int\limits_0^t e^{-\lambda_m (t-s)} \rho_0(x,s) ds,
\end{equation*}
$$
\lambda_m =  \lambda(1) + m.
$$
After substitution of $\rho_1(x,t)$  to the first equation in  \eqref{P0P1} we obtain the following evolution equation on $\rho_0$:
\begin{equation}\label{ro-0}
\partial_t \rho_0 = \Theta \cdot \nabla \nabla \rho_0 - b \cdot \nabla \rho_0 -  \lambda(0) \rho_0 +  \lambda(0) \lambda(1)  \int\limits_0^t e^{-\lambda_m (t-s)} \rho_0(x,s) ds + \lambda(1) e^{- \lambda_m t} \pi_1(x),
\end{equation}
with $\rho_0 (x,0) = \pi_0(x)$.

\medskip

{ Let us consider the stationary problem $L^* \, \rho = 0$ for the macroscopic model in $\Pi = \mathbb{T}^{d-1} \times \mathbb{R}^1_+$.
The equation on $\rho_0(x)$ takes the form
\begin{equation}\label{stationary}
\begin{array}{l} \displaystyle
\Theta \cdot \nabla \nabla \rho_0(x) -  b \cdot \nabla \rho_0(x) -  \lambda(0) \rho_0(x) +
 \frac{\lambda(0) \lambda(1)}{\lambda(1) + m } \ \rho_0(x) =
\\[4mm]  \displaystyle
= \Theta \cdot \nabla \nabla \rho_0(x) -  b \cdot \nabla \rho_0(x) -  \lambda(0) \frac{m}{\lambda(1) + m } \ \rho_0(x) = 0
\end{array}
\end{equation}
with a boundary conditions
\begin{equation}\label{BC}
\rho_0 (x|_{x_d =0}) = \varphi(x),  \quad \rho_0(x|_{x_d = \infty}) = 0,
\end{equation}
where $x_d$ is the direction of the drift. Here $\varphi(x)\ge 0$ is the profile of the initial concentration on the upper cross section.

Assuming 
that the initial profile $\varphi$ is a constant function
one can reduce the dimension in  problem \eqref{stationary}-\eqref{BC}  and obtain  a one-dimensional stationary problem that reads 
\begin{equation}\label{C}
\left\{
\begin{array}{l} \displaystyle
\theta \rho''_0 -  b \rho'_0 -  \varkappa \rho_0 =0, \quad \varkappa =  \lambda(0) \frac{m}{\lambda(1) + m }
\\[4mm]  \displaystyle
\rho_0(0)=1, \quad \rho_0(+\infty) =0.
\end{array}
\right.
\end{equation}
Thus, the rate of the purification process is equal to $ R_{pur}= \frac{1}{2\theta}\big(\sqrt{b^2+4\theta \varkappa} - b\big)$,
and  for sufficiently small $\theta$ we get $R_{pur} \approx \frac{\varkappa}{b}$.
}

\begin{remark}
If the astral set $A$ contains more that one points, i.e. $|A|=M>1$, then the kernel $K(t-s)$ in \eqref{ro-0}
is a linear combination of exponents $e^{-\varkappa_j (t-s)}$ with $\varkappa_j>0, \ j=1, \ldots, M$.
\end{remark}

\section{Appendix 1: proofs of the propositions}

\begin{proof}[Proof of Proposition \ref{Prop2}]

Considering the continuity of functions $f_j$ and the fact that $|w-z|\leq c\eps$,  we have
\begin{equation}\label{8}
(L_\varepsilon^0 F_\varepsilon^Q + V_\varepsilon F_\varepsilon)(z) = \sum_{j=1}^M\Big((T^0_\varepsilon - I) q_j (\frac{z}{\varepsilon})\Big) (f_0(z) - f_j(z)) +  \sum_{j=1}^M  \sum_{w \in \varepsilon \{ x_j \}^\sharp} v_\varepsilon(z, w) (f_j(z) - f_0(z)) + o(1),
\end{equation}
where $o(1)$ tends to 0 as $\varepsilon\to 0$.
From (\ref{7}) and (\ref{8}) we obtain the following system of uncoupled equations on the functions $q_j (\frac{z}{\varepsilon}), \; z \in \varepsilon B^\sharp,$ and constants $\alpha_{0 j}$:
$$
 \Big((T^0_\varepsilon - I) q_j (\frac{z}{\varepsilon})\Big) (f_0(z) - f_j(z)) +  \sum_{w \in \varepsilon \{ x_j \}^\sharp} v_\varepsilon(z, w) (f_j(z) - f_0(z)) = \alpha_{0j}  (f_j(z) - f_0(z)), \quad j = 1, \ldots, M.
$$
Then for every $j = 1, \ldots, M$, the function $q_j (\frac{z}{\varepsilon})$ satisfies the equation
\begin{equation}\label{9}
(T^0_\varepsilon - I) q_j (\frac{z}{\varepsilon}) =    \sum_{w \in \varepsilon \{ x_j \}^\sharp} v_\varepsilon(z, w)  - \alpha_{0j} {\bf 1_B}, \quad z \in \varepsilon B^{\sharp},
\end{equation}
which is of equivalent the following equation on  $ B^\sharp$:
\begin{equation}\label{P2_4}
(P^0 - I) q_j (x) =  \sum_{y \in  \{ x_j \}^\sharp}  v (x, y) - \alpha_{0j} {\bf 1_B}, \quad  x \in B^\sharp,
\end{equation}
where ${\bf 1_B}(x) \equiv 1 \; \forall x \in B^\sharp$, and $q_j(x)$ is $Y$-periodic.
Using Fredholm' alternative we conclude that the equation (\ref{P2_4}) has a unique solution if
$$
 \sum_{y \in  \{ x_j \}^\sharp}  v (x, y)  - \alpha_{0j} {\bf 1_B} \ \bot \mbox{ Ker } (P^0 - I)^\ast =  {\bf 1_B}.
$$
(The last relation follows from the irreducibility of $P^0$ on $B^\sharp$).
This condition implies the unique choice of constants $\alpha_{0j}$ given by formula \eqref{alpha_prior}, such that
the equation (\ref{P2_4}) has a unique bounded periodic solution $q_j(x), \ x \in B^\sharp$.
Proposition \ref{Prop2} is proved.
\end{proof}
\bigskip

\begin{proof}[Proof of Proposition \ref{Prop1}]

Using (\ref{FP}) we get for all  $x \in \varepsilon B^{\sharp}$:
\begin{equation}\label{P1_1}
\begin{array}{l} \displaystyle
L_\varepsilon^0 F_\varepsilon^P (x) =  \frac{1}{\varepsilon^2} (T_\varepsilon^0 - I) \left( f_0 (x) + \varepsilon \left(\nabla f_0 (x), h(\frac{x}{\varepsilon}) \right) \right) +  (T_\varepsilon^0 - I) \left(\nabla \nabla f_0 (x), g(\frac{x}{\varepsilon})\right)\\[4mm] \displaystyle
+  \frac{1}{\varepsilon} D_\varepsilon \left( f_0 (x) + \varepsilon (\nabla f_0 (x), h(\frac{x}{\varepsilon}) ) \right) +  O(\varepsilon).
\end{array}
\end{equation}
Then the vector function $h(\frac{x}{\varepsilon})$ can be found from the relation
\begin{equation}\label{P1_2}
\frac{1}{\varepsilon^2} (T_\varepsilon^0 - I) \left( f_0 (x) + \varepsilon \left(\nabla f_0 (x), h(\frac{x}{\varepsilon}) \right) \right) = O(1).
\end{equation}
Using notation \eqref{p0-sym-bis} we can write $T^0_\varepsilon f (x)$  as follows:
\begin{equation}\label{Pxi1}
(T^0_\varepsilon f)(x) \ = \ \sum_{\xi} p_\xi (\frac{x}{\varepsilon}) f(x+ \varepsilon \xi), \quad x \in \varepsilon B^{\sharp}.
\end{equation}
It follows from (\ref{Pxi1}) that the left-hand side of (\ref{P1_2}) takes the form:
\begin{equation}\label{P1_3}
\frac{1}{\varepsilon^2} \sum_{\xi} p_\xi (\frac{x}{\varepsilon}) \left( f_0 (x + \varepsilon \xi) - f_0(x) \right)  +
\frac{1}{\varepsilon} \sum_{\xi} p_\xi (\frac{x}{\varepsilon})
\left( \left(\nabla f_0 (x + \varepsilon \xi), h(\frac{x}{\varepsilon}+\xi) \right) - \left(\nabla f_0 (x), h(\frac{x}{\varepsilon}) \right)  \right)
\end{equation}
$$
= \frac{1}{\varepsilon} \sum_{\xi} p_\xi (\frac{x}{\varepsilon}) \left( \nabla f_0 (x), \xi \right)  +
\frac{1}{\varepsilon} \sum_{\xi } p_\xi (\frac{x}{\varepsilon})
\left( \nabla f_0 (x), h(\frac{x}{\varepsilon}+\xi) - h(\frac{x}{\varepsilon}) \right)  + O(1)
$$
$$
= \frac{1}{\varepsilon}   \left( \nabla f_0 (x), \sum_{\xi} p_\xi (\frac{x}{\varepsilon}) \left(  \xi + ( h(\frac{x}{\varepsilon}+\xi) - h(\frac{x}{\varepsilon}) \right) \right)  + O(1).
$$
Thus the periodic vector function $h(x)$ is taken as a solution of the equation
\begin{equation}\label{P1_4}
(P^0 - I) \left( l (x) + h(x) \right) =0, \quad x \in B^\sharp,
\end{equation}
 where $l(x)=x$ is the linear function. The solvability condition for equation (\ref{P1_4}) reads
$$
((P^0 - I)l,  \mbox{ Ker } (P^0 - I)^\ast) = ((P^0 - I) l ,\ {\bf 1_B}) = \sum_{x \in B} \sum_{\xi} p_\xi(x) \xi = 0.
$$
Since $p_{\xi}(x) = p_{-\xi}(x+\xi)$, this condition holds true, which implies the existence of the unique, up to an additive constant,  periodic solution $h(x)$ of equation (\ref{P1_4}).

\medskip

We follow the similar reasoning to find an equation for the periodic matrix function $g(x), \ x \in B^\sharp$. We will also obtain below expressions for effective matrix $\Theta$ and drift $b$.

Collecting in (\ref{P1_1}) all terms of the order $O(1)$, using  relation (\ref{P1_4}) on the  function $h(x)$ and relation \eqref{v0} on matrix $D$ we get:
$$
\frac{1}{\varepsilon^2} \sum_{\xi} p_\xi (\frac{x}{\varepsilon}) \left( f_0 (x+ \varepsilon \xi) - f_0(x) \right)  +
\frac{1}{\varepsilon} \sum_{\xi } p_\xi (\frac{x}{\varepsilon})
\left( \left(\nabla f_0 (x+\varepsilon \xi), h(\frac{x}{\varepsilon}+\xi) \right) - \left(\nabla f_0 (x), h(\frac{x}{\varepsilon}) \right)  \right)
$$
$$
+ \sum_{\xi } p_\xi (\frac{x}{\varepsilon}) \left( \left( \nabla \nabla f_0 (x + \varepsilon \xi), g(\frac{x}{\varepsilon} + \xi) \right) -
 \left( \nabla \nabla f_0 (x), g(\frac{x}{\varepsilon} ) \right) \right)
$$
$$
+ \frac{1}{\varepsilon}
 \sum_{\xi} d_\xi (\frac{x}{\varepsilon}) f_0 (x+ \varepsilon \xi) +   \sum_{\xi} d_\xi (\frac{x}{\varepsilon})  \left(\nabla f_0 (x+ \varepsilon \xi), h(\frac{x}{\varepsilon}+ \xi) \right)  + O(\varepsilon)
$$
$$
= \frac{1}{\varepsilon}   \left( \nabla f_0 (x), \sum_{\xi} p_\xi (\frac{x}{\varepsilon}) \left(  \xi + ( h(\frac{x}{\varepsilon}+\xi) - h(\frac{x}{\varepsilon}) \right) \right)
$$
$$
+ \left( \nabla \nabla f_0 (x), \sum_{\xi} p_\xi (\frac{x}{\varepsilon}) \left(  \frac12 \xi \otimes \xi + \xi \otimes  h(\frac{x}{\varepsilon}+\xi)  + (g(\frac{x}{\varepsilon} + \xi) - g(\frac{x}{\varepsilon} ) ) \right) \right)
$$
$$
+  \left( \nabla f_0 (x), \sum_{\xi} d_\xi (\frac{x}{\varepsilon}) (\xi +  h(\frac{x}{\varepsilon}+ \xi)) \right)
+ O(\varepsilon)
$$
\begin{equation}\label{P1_5}
= \left( \nabla \nabla f_0 (z), \sum_{\xi \in \Lambda_{\frac{z}{\varepsilon}}} p_\xi (\frac{z}{\varepsilon}) \left(  \frac12 \xi \otimes \xi + \xi \otimes  h(\frac{z}{\varepsilon}+\xi) \right)  + (P^0 -I) g(\frac{z}{\varepsilon})   \right)
\end{equation}
$$
+ \left( \nabla f_0 (x), \sum_{\xi} d_\xi (\frac{x}{\varepsilon}) (\xi +  h(\frac{x}{\varepsilon}+ \xi)) \right) + O(\varepsilon).
$$
Let $\frac{x}{\varepsilon} = y \in B$, and denote by $\Phi(h)$ the following matrix and vector functions
\begin{equation}\label{P1_6}
\Phi(h)(y) =   \frac12 \sum_{\xi \in \Lambda_y} p_\xi (y)\, \xi \otimes \xi + \sum_{\xi \in \Lambda_y} p_\xi (y)\, \xi \otimes h(y+ \xi),
\end{equation}
\begin{equation}\label{b-drift}
b(y) = \sum_{\xi} d_\xi (y) (\xi +  h(y + \xi)) =  \sum_{\xi} d (y, y+ \xi) (\xi +  h(y + \xi)),  \quad y \in B.
\end{equation}
In order to ensure the convergence in (\ref{PP1}) we should find a constant matrix $\Theta$, a periodic matrix
function $g(y)$ and a constant vector $b $ such that
\begin{equation}\label{P1_7}
\Phi(h)_{km}(y) + (P^0 - I) g_{km}(y) = \Theta_{km},
\end{equation}
\begin{equation}\label{P1_71}
( b(y) - b, {\bf 1_B}) = \sum_{y \in B} \sum_{\xi} d (y, y+\xi) (\xi +  h(y + \xi)) - b \,|B| = 0.
\end{equation}
The latter equation implies that
\begin{equation}\label{b}
b = \frac{1}{|B|} \sum_{y \in B} \sum_{\xi } d(y, y + \xi)( \xi  +  h(y + \xi)).
\end{equation}
The solvability condition for (\ref{P1_7}) reads
$$
(-\Phi(h)_{km} + \Theta_{km}, \ Ker (P^0 - I)^\ast) = (-\Phi(h)_{km} + \Theta_{km} ,\ {\bf 1_B})  = 0.
$$
Thus $\Theta$ is uniquely defined as follows:
\begin{equation}\label{theta}
\Theta_{k m} = \frac{1}{|B|} \sum_{y \in B} \Phi_{k m}(h) (y), \quad \mbox{ where } \quad  \Phi(h)(y) = \sum_{\xi} p_\xi (y)\, \xi \otimes \left( \frac12  \xi  +  h(y + \xi) \right),
\end{equation}
and $g(y)$ is a solution of equation \eqref{P1_7}. This solution is uniquely defined up to a constant matrix.
We notice that the matrix $\Theta$ defined by \eqref{theta} is  positive definite, i.e. $(\Theta \eta, \eta)>0 \; \forall \eta \neq 0$. The proof is given in \cite{PZh}.

This complete the proof of Proposition \ref{Prop1}.
\end{proof}

\section{Appendix 2. One example with calculation of effective parameters.}

In this section we consider one example of a model at the microscopic scale and calculate for this example the effective parameters $\Theta$ and $ b$ of the limit model that in particular used in the description of the stationary regime, see equation \eqref{stationary}.

Let the periodicity cell $Y$ be a square $3 \times 3$ of the two-dimensional lattice $\mathbb Z^2$, $A$ is the one-point subset of $Y$ located at the center of $Y$, and $B = Y \backslash A$, t.e. $|B| = 8$. Let us renumber the elements from $B$ in accordance with their position on the cell $Y$:
$$
\boxed{
\begin{array}{ccc}
s_1& s_2& s_3\\
s_4 & \bullet &  s_5\\
s_6 & s_7 & s_8
\end{array}
}
$$
We define the symmetric matrix $P_0|_{B^\sharp} = \{ p_0(x,y), \ x,y \in B^\sharp \}$ describing the free moving of the water outside of cleaning elements as follows:
\begin{itemize}
\item[\bf--] $ \ p_0(x,x \pm e_1) = p_0(x,x \pm e_2) = \frac14$, if $x \in \{s_1, s_3, s_6, s_8\}$;
\item[\bf--] $ \ p_0(x, x \pm e_1) = \frac14, \; p_0(x,x + e_2) = \frac12$, if $x = \{s_2\}$;
\item[\bf--] $ \ p_0(x, x \pm e_2) = \frac14, \; p_0(x,x - e_1) = \frac12$, if $x = \{s_4\}$;
\item[\bf--]  $ \ p_0(x, x \pm e_2) = \frac14, \; p_0(x,x + e_1) = \frac12$, if $x = \{s_5\}$;
\item[\bf--]  $ \ p_0(x, x \pm e_1) = \frac14, \; p_0(x,x - e_2) = \frac12$, if $x = \{s_7\}$.
\end{itemize}
Other elements of the matrix $P_0|_{B^\sharp}$ equal to 0.

Next the matrix $D$ is introduced, which determines a nonzero drift (of the order $\varepsilon$) in the microscopic scale model.
Taking into account relation \eqref{v0} we consider the following elements of $D = \{ d(x,y) \}, \ x,y \in B^\sharp$:
\begin{equation}\label{Ditem}
d(x,x \pm e_2) =  \mp K, \quad \mbox{if } \; x \notin \{s_2, s_7\}, \qquad \mbox{ and } \; d(x,y) =0, \quad \mbox{ otherwise}.
\end{equation}
Thus, the model at the microscopic scale on the component $B^\sharp$ complementary to the astral component is completely defined.

As follows from results of the previous section the effective parameters, the matrix $\Theta$ and the vector $ b$,  are given by \eqref{theta} and \eqref{b-drift} respectively. Since
$$
p_\xi(x) = p_0(x,x+\xi), \quad d_\xi(x) = d(x,x+\xi), \qquad x,\ x+\xi \in B^\sharp,
$$
have been already introduced above in this section, we have to find the vector function $h(x) = \{ h(x),\ x \in B \}$. This function is called corrector, and it is taken as a solution of the equation \eqref{P1_4}. Thus, in our example the corrector is the same as a set of eight vectors
$$
h^B = \{ h(s_1) \in \mathbb{Z}^2, \  h(s_2) \in \mathbb{Z}^2, \ldots, \ h(s_8) \in \mathbb{Z}^2 \}, \quad B=\{s_1, \ldots, s_8\}.
$$

We explain now how to find this vector function.
It is worth noticing that \eqref{P1_4} is a system of uncoupled equations, and we can solve it for each coordinate separately. To find the vector of the first coordinates $h^B_1 = \{h_1(s_1), h_1(s_2), \ldots, h_1(s_8) \}$ we take in \eqref{P1_4} $l_1(x) = (x_1, 0)$ and rewrite \eqref{P1_4} in the following way:
\begin{equation}\label{P1_4bis}
(P^0 - I) h_1(x) = - (P^0 - I)  l_1(x) =: g_1(x), \quad x \in B^\sharp,
\end{equation}
where $g_1(x) =   - (P^0 - I)  l_1(x) = (0,0,0,1/2, -1/2,0,0,0)$ is the function defined on $B$ and periodic on $\mathbb{Z}^2$. It can be represented as follows:
$$
\boxed{
\begin{array}{ccc}
0& 0& 0\\[1.5mm]
1/2 & \bullet & \!\! -1/2\\[1.5mm]
0 & 0 & 0
\end{array}
}
$$
Thus, it follows from \eqref{P1_4bis} that the vector of the first coordinates $h^B_1$ is equal to
\begin{equation}\label{hB-1}
h^B_1 = (P^0 - I)^{-1} g_1.
\end{equation}
Similarly, setting $l_2(x) = (0, x_2)$, we obtain the vector of the second coordinates $h^B_2$ of the set $h^B$
by the formula
\begin{equation}\label{hB-2}
h^B_2 = (P^0 - I)^{-1} g_2, \qquad \mbox{where } \; g_2 = -(P^0 - I)l_2.
\end{equation}
Then using \eqref{hB-1} - \eqref{hB-2} one can find the vector $b$ by the formula \eqref{b-drift}. In this example we get
$$
b =  \sum_{x \in B}\sum_{\xi= \pm e_1, \pm e_2} d (x, x+ \xi) (\xi +  h(x + \xi)) = (0, - \frac32 K),
$$
where $K$ is the same constant as in \eqref{Ditem}. Here we used that by the periodicity assumption for any $x \in B$ and any $\xi \in \mathbb{Z}^2$
$$
(x + \xi)_{{\rm mod} B} = \hat x, \quad \mbox{with some } \; \hat x \in B,
$$
so that $h(x + \xi) = h(\hat x)$.

Finally the matrix $\Theta$ of order $2 \times 2$ can be found by the formula \eqref{theta}, where each term in the sum \eqref{theta}:
$$
\Phi(h)(x) = \sum_{\xi = \pm e_1, \pm e_2} p_\xi (x)\, \xi \otimes \left( \frac12  \xi  +  h(x + \xi) \right), \quad x \in B,
$$
is determined in terms of the corrector $h(x)$ and the matrix elements of $P_0|_{B^\sharp}$. We used here that
$$
\left(
\begin{array}{c}
a_1 \\ a_2
\end{array}
 \right) \otimes \left(
 \begin{array}{c}
b_1 \\ b_2
\end{array} \right) =
\left(
\begin{array}{cc}
a_1 b_1 & a_1 b_2 \\ a_2 b_1 & a_2 b_2
\end{array}
 \right)
$$

\end{document}